\documentclass[11pt,letterpaper]{article}

\usepackage{color}              
\usepackage[suppress]{color-edits}
\addauthor{et}{black}
\addauthor{etn}{red}
\addauthor{hbn}{magenta}
\addauthor{hb}{black}
\usepackage{hyperref}
\usepackage[numbers]{natbib}
\usepackage{xspace}
\usepackage{xcolor}
\usepackage{graphicx}
\usepackage{amsmath,amssymb,amsthm}
\usepackage[margin=1in]{geometry}
\usepackage{tikz}
\usetikzlibrary{chains,arrows}
\usepackage{todonotes}

\theoremstyle{plain}

\newtheorem{theorem}{Theorem}[section]
\newtheorem{lemma}[theorem]{Lemma}
\newtheorem{corollary}[theorem]{Corollary}
\theoremstyle{definition}
\newtheorem{example}{Example}

\newtheorem{property}[]{Property}

\def\squareforqed{\hbox{\rule{2.5mm}{2.5mm}}}

\def\QED{\ifmmode\squareforqed 
  \else{\nobreak\hfil   
    \penalty50                 
    \hskip1em                  
    \null                      
    \nobreak                   
    \hfil                      
    \squareforqed              
    \parfillskip=0pt           
    \finalhyphendemerits=0     
    \endgraf}                  
  \fi}

\def\blksquare{\rule{2mm}{2mm}}
\def\qedsymbol{\blksquare}
\newcommand{\bg}[1]{\medskip\noindent{\bf #1}}
\newcommand{\ed}{{\hfill\qedsymbol}\medskip}

\newenvironment{proofof}[1]{{\it{Proof of #1 : }}}{\ed}

\newtheorem{thm}{Theorem}

\newcommand{\FullversionOmit}[1]{} 

\newtheorem{claim}[theorem]{Claim}

\newtheorem{definition}[theorem]{Definition}

\newcommand{\fullversion}[1]{}

\begin{document}

\title{Effect of Strategic Grading and Early Offers in Matching Markets
}

\author{%
	Hedyeh Beyhaghi
\thanks{{\tt hedyeh@cs.cornell.edu}, Dept of Computer
Science, Cornell University. Supported in part by NSF grants CCF-0910940.}
	\and
Nishanth Dikkala
\thanks{ {\tt nishanthd@csail.mit.edu}, Dept of Electrical Engineering and Computer Science, Massachusetts Institute of Technology.
research was done while visiting Cornell university, supported by NSF grant CCF-0910940.}
	\and
	\'{E}va Tardos\thanks{ {\tt eva@cs.cornell.edu}, Dept of Computer
Science, Cornell University. Supported in part by NSF grants CCF-0910940 and CCR-1215994), ONR grant N00014-08-1-0031,  a Yahoo!~Research Alliance Grant, and
a Google Research Grant.}
}

\date{\today}

\maketitle

\thispagestyle{empty}
\setcounter{page}{0}

\begin{abstract}
Strategic suppression of grades, as well as early offers and contracts, are well-known phenomena in the matching process where graduating students apply to jobs or further education. In this paper, we consider a game theoretic model of these phenomena introduced by Ostrovsky and Schwarz, and study the loss in social welfare resulting from strategic behavior of the schools, employers, and students. We model grading of students as a game where schools suppress grades in order to improve their students' placements. We also consider the quality loss due to unraveling of the matching market, the strategic behavior of students and employers in offering early contracts with the goal to improve the quality. Our goal is to evaluate if strategic grading or unraveling of the market (or a combination of the two) can cause significant welfare loss compared to the optimal assignment of students to jobs. To measure welfare of the assignment, we assume that welfare resulting from a job -- student pair is a separable and monotone function of student ability and the quality of the jobs. Assuming uniform student quality distribution, we  show that the quality loss from the above strategic manipulation is bounded by at most a factor of 2, and give improved bounds for some special cases of welfare functions.
\end{abstract}

\newpage

\section{Introduction}
We consider the effect of strategic behavior in matching markets as school graduates get assigned to jobs (or to further education) reacting to multiple incentives:
\begin{itemize}
 \item Companies want to hire the best students,
 \item Students want to get the best jobs,
 \item Schools want to help their graduating students get great jobs.
\end{itemize}
Here, we consider simple games that model ways schools, students, and employers can respond to these incentives.  Before we formally introduce our games, we briefly review the effect of these incentives on the various placement markets.

To help the placement of their students, schools often like suppressing grades, especially grades of their top performing students. The Stanford Graduate School of Business (GSB) has a policy for not reporting grades, many other schools suppress A+ grades on the transcripts, replacing all with $A$, or suppress +/- signs altogether. Some high schools report class-rank, but many other schools refuse to do so. An important reason for suppressing grades is the desire to have better placements for all students in the school, not only the top performing students. Suppressing grades is an explicit policy of many high schools and universities, but the same effect can also be achieved by allowing grade inflation, or by not having clear grading guidelines. When a large fraction of the class receives a grade of A, the expressiveness of grades suffer, effectively creating the same effect as suppressing grades by other schools. Similarly, randomness in how grades are assigned also decreases the information content of the transcripts.

Students and companies or schools of further education also behave strategically, acting to get better jobs or to improve the quality of students they can hire or attract. One tool in this area is making early offers, referred to as the unraveling of the matching market. Companies at times make offers to students quite a bit before they graduate, based on transcripts with significant amount of course work, and hence important information is still missing. Students often accept these early offers, or even apply for them. Such early offers are common for students in Business and Law schools. There is also a similar trend in students applying to college under early decision programs offered by many schools.

Ostrovsky and Schwarz \cite{Ostrovsky_SchwarzAEJ10} introduced the model of unraveling of matching markets and the game of suppressing grades that we study in this paper. They offer a model of jobs, students, and grading to study these phenomena.  They view grading as a form of signaling about the student quality by the school,  and assume that all participants are risk neutral (so employers aim to maximize the expected quality of the students they hire). They model grading as a game among schools, and assume that each school aims to release information about their students with the goal to help them get the best jobs. Their main result is that (under mild assumptions) at a Nash equilibrium of this game, schools disclose the right amount of information, so that students and employers will not find it profitable to contract early.

We also consider the partial unraveling of the matching market: early contracting. While Ostrovsky and Schwarz \cite{Ostrovsky_SchwarzAEJ10} show that early contracting is not advantageous under an equilibrium grading policy of their information disclosure game, we observe that early contracting is increasingly common. Schools do aim to optimize the placements of their students, but we believe that they do not fully optimize grading. The pervasiveness of early contracting does suggest that the information released in grades is not at the equilibrium of the disclosure game. Exact optimization of grades at the full generality proposed by the model of Ostrovsky and Schwarz is also not feasible, or even advisable as grades play many roles, including motivating the students \cite{Dubey}.

In this paper, we consider the effect of such strategic actions on the social welfare, the overall quality of the assignment. To do this, we need to model the way that placement of a student with ability $a$ in a job with quality $q$, will contribute to welfare. We assume that  the resulting welfare is a monotone increasing function of both $a$ and $q$, and the effects of these two contributing factors are separable. Concretely,
we assume that the resulting welfare is expressed as $f(q)g(a)$ with both $f$ and $g$ nondecreasing, and $g$ also concave.

In sections \ref{sec:static} and \ref{sec:early} we assume that welfare is expressed simply as $aq$ without functions $f$ or $g$. Effectively this assumption means that we identify the quality of a job $q$ with its value $f(q)$ to contribute to welfare, and similarly identify the ability $a$ of a student with his or her value $g(a)$.
Alternately, we can think of this special case as a change in schools' and employers' objectives, in this model an employer is aiming to optimize not the expected ability $E(a)$ of a student hired, but rather the student's expected contribution to welfare $E(g(a))$, and similarly, assume that schools evaluate the placements of their students
by the average ability $E(f(q))$ of these jobs to contribute to social welfare.\\

\noindent{\bf Our results}\\
\etedit{We analyze the quality loss in assignments in two different forms of strategic behavior: strategic grading by schools, and alternately, early contracting by employers.} 

In Section \ref{sec:general} we consider a very general model where \etedit{with no assumption on schools grading policies}, where students and employers respond by early contracting for an arbitrary grading policy by schools. Following Ostrovsky and Schwarz \cite{Ostrovsky_SchwarzAEJ10} we use a continuous model, assuming that there are infinitely many students and each school is infinitesimally small, see the formal definition in section \ref{sec:prelim}. To simplify our model we consider a two stage game where in stage one some student-employer pairs can agree on early contracts. In the second stage grades are released based on each school's grading policy, and the remaining students and jobs are matched based on their grades. In this two stage game, prospective employers have to make decisions about early offers without any grade information about the students, solely based on the school that the student attends\etedit{, where we assume that the distribution of students in each school is public knowledge}. We think of grades as a form of signaling, and identify all grades with the expected ability of the group of students who receive that grade. We show that if the abilities of students are uniformly distributed, the resulting quality of the assignment at an equilibrium of the matching game with early contracting is at most a factor of 2 away from the best possible with any grading policy by the schools. 
This bound is best possible without further assumption on grading, as this is the assignment resulting when all schools have identical student populations, and all schools refuse to release grades.

Next we consider the decrease of the quality of the matching resulting from the fully strategic grading used by a school, i.e., the price of anarchy of the information disclosure game of Ostrovsky and Schwarz \cite{Ostrovsky_SchwarzAEJ10}. Recall that fully strategic grading eliminates the incentive for early contracting. Our general results from Section \ref{sec:general} imply a price of anarchy bound of at most 2. We focus on the case when welfare is measured as $aq$ (alternately, assuming that employers and schools evaluate students and jobs respectively with their ability to contribute to social welfare). We show that in this special case, the price of anarchy of the strategic grading game is bounded by 1.36.
In the appendix, we also give a 1.07 lower bound on the price of anarchy, showing that the quality of matching can indeed degrade by a constant factor due to strategic grading even in this special case.

In Sections \ref{sec:early} we consider the  quality of the matching resulting from early contracting in isolation, assuming fully informative grading, and again focusing on the case when welfare is measured as $aq$. Our general results from Section \ref{sec:general} imply a price of anarchy bound 2. In the appendix, we give a 1.22 lower bound. In the special case, when all schools have identical (and uniform) distribution of student abilities, we show that the price of anarchy is bounded by 4/3. In the appendix, we also give a 1.11 lower bound for this special case.\\

\noindent{\bf Related work}\\
The model of unraveling of matching markets and the game of suppressing grades we consider in this paper, was introduced by Ostrovsky and Schwarz \cite{Ostrovsky_SchwarzAEJ10}. There are a number of papers reporting unraveling phenomena in various matching markets  from selective colleges \cite{Avery_etal03}, to the market for law clerks \cite{Avery_etal01} that report that interviews for law clerk postings are held almost two full years prior to graduation.

Roth and Xing \cite{roth1994jumping} offer a number of examples of matching markets that are unraveling, and also provide a framework for modeling unraveling.
They assume that both sides are risk-neutral agents, and show that unraveling may happen even if the final matching is stable. They also show that unraveling may lead to inefficiency in the matching.
Chan, Hao and Suen \cite{ChanLiSuen} offer a game theoretic model of grade inflation, a different game where schools use grades to improve the placement of their students.

Grading is a form of information disclosure by schools.  Bergemann and Pesendorfer \cite{bergemann2007information} and recently Dughmi \cite{DughmiFOCS} studied suppression of information in a very different setting. They designed optimal single-seller, single-object auctions assuming that seller can hide information from the bidders about their valuation. The method for suppressing information that they use is the same as the method used in this paper.


\section{Preliminaries}\label{sec:prelim}
In this paper, we consider the matching problem for students finishing school and getting placed for jobs or higher education. We assume that schools give students transcripts, which are used by employers in their hiring decisions. We assume that schools conduct exams that measure students' abilities perfectly, but do not require that the transcripts are fully informative, as less informative transcripts may improve the placement of the students. \etedit{We assume that the distribution of students in each school is public knowledge, as well as the distribution of student abilities with a given grade in each school, i.e., the grading policy of the school is public knowledge, but employers have no other way of measuring the quality of the students.}

To model the key aspects of this matching process, we assume that each student has a true ability, a single real number in the range $[a_L,a_H]$, where a student with higher ability $a$ is more desirable for all jobs. We use a model with a continuum of students abilities, and assume throughout the paper that the overall distribution of students abilities is uniform. We also assume that the desirability of each position is single dimensional, described by a number $q \in [q_L,q_H]$, denoted as job quality, and is common knowledge. All students prefer jobs with higher quality $q$.  The distribution $\mu(.)$ of position desirabilities is continuous, exogenous, commonly known and has positive density on $[q_L, q_H]$, but not necessarily uniform.
Finally, for simplicity of notation, we assume without loss of generality that the mass of positions is equal to the mass of students.

\hbdelete{If the students' true ability is known, the resulting mapping is a unique matching (up to permutations of equally desirable positions) between the abilities of the students and jobs desirabilities,}
\hbedit{If the students' true ability is known, based on students ranking on one side and jobs ranking on the other side, there is a unique stable matching (up to permutations of equally desirable positions) between students and jobs,}
where higher abilities are mapped to more desirable jobs and vice versa. We will use $Q_T$ to denote this function mapping students abilities to jobs qualities.

\begin{definition}
Function $Q_T(.)$ or \emph{the truthful mapping},
is the mapping of student ability to  desirability of the assigned job \hbedit{based on the stable matching,} when all students abilities are known. This mapping is increasing in student ability.
\end{definition}

\etedit{We will think of $Q_T(.)$ as the ideal mapping of students to jobs. With our assumption that abilities of students, and the qualities of jobs linearly ordered, this would be the resulting assignment of students to jobs with if all information would be available:  as employers prefer better students, and all students prefer better jobs.} We will need to make the technical assumption that $Q_T(.)$ does not switch between convex and concave parts infinitely many times.

In this paper we will consider different games that modify this matching, due to strategic grading and/or early contracting. \etedit{The grading policy of a school assigns grades to all students. We will think of a grade as a signal of the student's ability. Recall that we assumed that the distribution of abilities of students with a given grade is public knowledge. We will further assume that the employers aim to hire a student with maximum expected ability (that is, they are risk neutral). With this in mind, we can identify a grade (from a given school) with the expected ability of the students receiving this grade. If employers have access to grades, and are risk neutral, the resulting mapping of students to jobs is based on the grades (not directly abilities) and monotone in the expected ability associated with the grades. }

\begin{example}
Suppose $\frac{1}{3}$ of schools have uniform distribution of abilities $[0,\frac{2}{3}]$. Other schools have half of their students uniform from $[0,\frac{2}{3}]$, and other half uniform from $[\frac{2}{3},1]$. Note this is a uniform aggregate distribution of  abilities. Suppose there is the same mass of jobs as students, where $\frac{1}{3}$ of them are distributed uniformly on $[0,0.5]$ and other $\frac{2}{3}$ are distributed uniformly on $[0.5,1]$.
If the schools reveal true abilities of students the mapping between expected abilities and jobs desirabilities is as follows. (Figure \ref{fig:example})
\begin{figure}
\centering
\includegraphics[height=3cm]{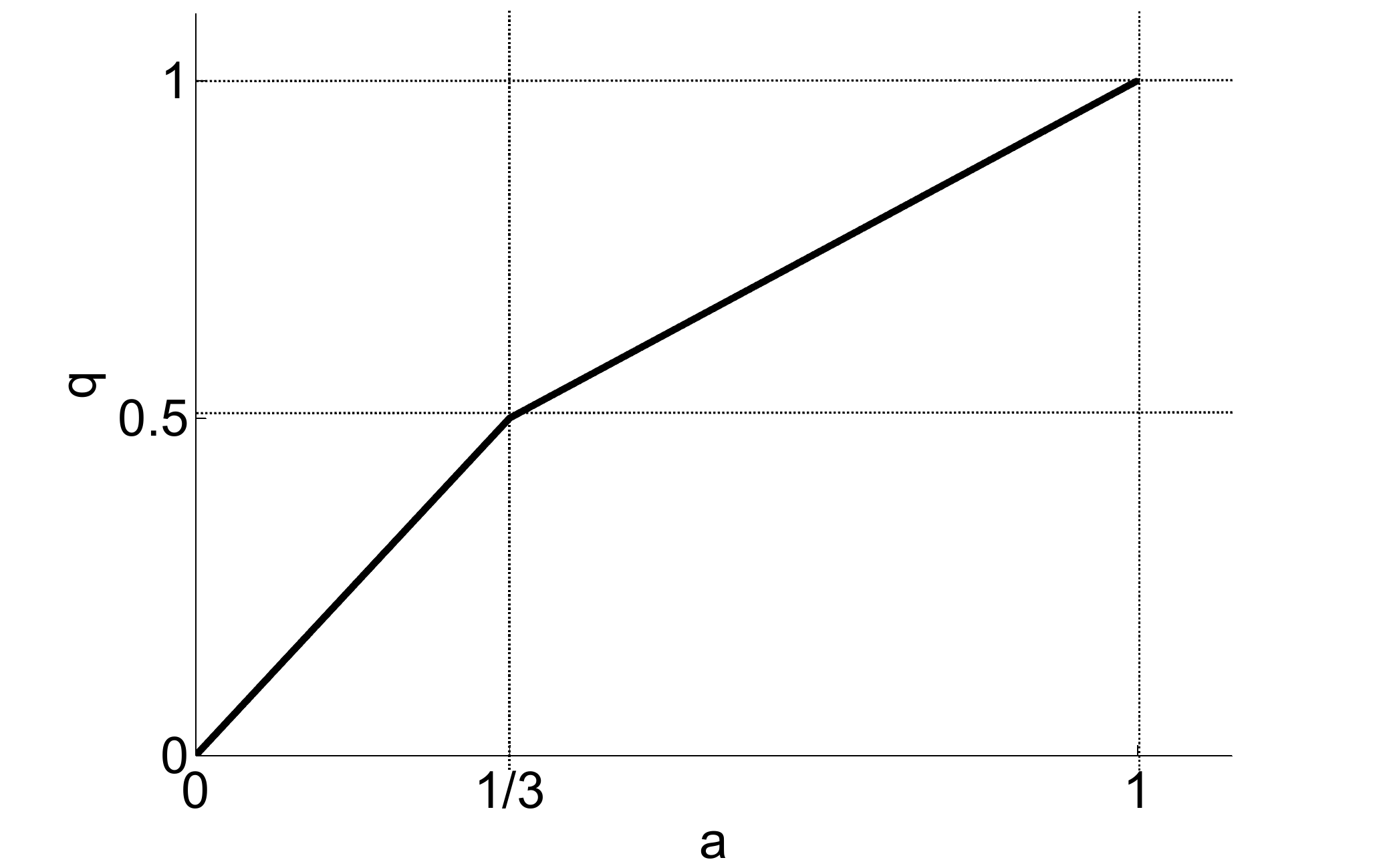}
\caption{The matching between expected abilities to job qualities. The horizontal axis ($\hat{a}$) shows the expected ability and the vertical axis ($q$) shows the quality.}
\label{fig:example}
\end{figure}
\begin{equation*}
Q(\hat{a}) =
\begin{cases}
\frac{3\hat{a}}{2} & \text{for $\hat{a} \leq \frac{1}{3}$}\\
\frac{1}{4}+\frac{3\hat{a}}{4} & \text{for $\hat{a} \geq \frac{1}{3}$}
\end{cases}
\end{equation*}

Any school that has students on the whole range, can improve its payoff by suppressing information. If they announce the same expected ability for all their students, their average job placement improves from $\approx 0.66$ to $0.75$.
\end{example}


Each game will result in a (probabilistic) mapping of students to jobs at equilibrium. For such a mapping $Q$, we will evaluate the social welfare of the resulting assignment. We define the social welfare of a mapping to be the sum of the \emph{value of interaction} of job/ability pairs. The value of interaction should be increasing both in student ability and in job quality. We assume the resulting welfare is a separable function of these two variables, and will use the product $f(q)g(a)$ to measure the social value of a student with ability $a$ being matched to a job of quality $q$ for  increasing functions $f$ and $g$. 



\begin{definition} For a mapping $Q$ from student expected abilities to jobs qualities, we will think of $Q^{-1}(.)$ as a randomized function mapping job qualities to the ability of the student assigned. We measure the resulting social welfare by
$\int_{q_L}^{q_H}f(q)g(Q^{-1}(q))d\mu(q)= \int_{q_L}^{q_H}f(q)g(Q^{-1}(q))\mu'(q) dq$.
\end{definition}

We will consider school's grading. Grades and the school a student attends, reveal some information about the student's ability, but this information is typically not completely informative, as different ability students can get the same grade. For notational simplicity, we will identify grades from each school with the expected ability that such a grade reveals. Now we can use  $G(.)$ to denote the cumulative distribution of expected abilities revealed by grading, i.e., $G(x)$ is the probability that a randomly selected student has a grade that corresponds to expected ability at most $x$. Using this notation, the integral of social welfare can also be written
as $\int_{a_L}^{a_H}g(a)f(Q(a))dG(a)=\int_{a_L}^{a_H}g(a)f(Q(a))G'(a)da$. Recall that we assume that the distribution for the truthful grading is uniform, so the welfare of the truthful assignment is $\int_{a_L}^{a_H}g(a)f(Q_T(a))da$.

Among all possible mappings $Q$, the mapping $Q_T$ resulting from fully informative\etedit{, truthful} transcripts is the one with maximum social welfare, since the highest ability person takes the highest desirability position and so on.
\begin{lemma}\label{lm:qt}
The maximum social welfare is obtained by the mapping $Q_T$ of students to jobs \etedit{based on the true abilities.}
\end{lemma}

\begin{proof}
Function $Q_T$ is the unique increasing function mapping student abilities to job qualities. Consider any other function $Q'$. There exist two pairs of student-job, $(a,q=Q'(a))$, $(a',q'=Q'(a'))$ such that $a<a'$ and $q'>q$. Since functions $f$ and $g$ are increasing, $f(a)g(q)+f(a')g(q') \leq f(a)g(q')+f(a')g(q)$. Therefore no other functions maximizes social welfare.
\end{proof}

\begin{definition}
For a game of matching students to jobs, the price of anarchy is the loss of welfare in equilibrium, defined as the ratio of the maximum possible welfare to the welfare of the resulting assignment:
$$\max_Q \frac{\int_{q_L}^{q_H}f(q)g(Q_T^{-1}(q))\mu'(q) dq}{\int_{q_L}^{q_H}f(q)g(Q^{-1}(q))\mu'(q) dq}$$
where the maximum ranges over all equilibria $Q$ of the game.
\end{definition}

An important special case of our welfare function is measuring welfare simply by $aq$, effectively identifying the student's quality with his/her ability $g(a)$ to contribute social welfare, and assuming that the desirability of a job $q$ is proportional to its ability $f(q)$ to contribute to social welfare. In Sections \ref{sec:static} and \ref{sec:early} we consider this special case, and continue to assume that the students abilities (that is now their ability to contribute to society) is uniformly distributed.

\paragraph{\bf The schools}
We model each school as infinitesimally small, and assume that there are infinitely many schools. Different schools may have different student ability distributions (we assumed that the overall distribution of students abilities is uniform, but different schools can be better or worse, and hence have higher or lower ability students). We will assume that there are only finitely many different types of schools.

\paragraph{\bf Range of students abilities and jobs qualities}
To simplify the notation, we will assume without loss of generality that students abilities, as well as jobs qualities are in the range $[0,1]$. Note that this shifting the interval of ranges, and resealing can only increase the price of anarchy.


\begin{lemma}\label{pr:reverse}
When the distribution of students abilities is uniform on $[0,1]$ and jobs qualities interval is $[0,1]$, the truthful mapping $Q_T$ is the inverse of jobs distribution $\mu$, that is $Q_T^{-1}(q)=\mu(q)$
\end{lemma}

\begin{proof}
Suppose $Q_T(a)=q$, since $Q_T$ is increasing in $a$, $[0,a]$ has been mapped to $[0,q]$. So the number of students in $[0,a]$, are the same as number of jobs in $[0,q]$. Due to assumption of uniform abilities, the number of people in $[0,a]$ is $a$. Thus the number of jobs in $[0,q]$ is $\mu(q)$ and $Q_T^{-1}(q)=\mu(q)$, as claimed.
\end{proof}

\paragraph{\bf The grading game}
In Sections \ref{sec:general} and \ref{sec:static} we will consider the schools' strategic behavior in suppressing grades. We view grades as signals about the student's ability. Schools may suppress the exact abilities of their students by assigning identical grades to students of different abilities. We assume that all jobs are risk neutral, and prefer higher average ability students independent of the distribution of abilities with a given mean.

\vskip 0.1in
\paragraph{\bf Grading strategies}
We assume that for each school, the grade is a signal, and the average ability of the set of students with a given grade becomes public knowledge. Different school may use different grading scales, such as grading with the range 1-5, 1-4, 1-4.3, 1-20, 1-10 used by different universities. In deciding which student to hire, an employee will consider the grade, and what the expected ability of a student is with this school.  If a grading policy is stable, which is the subject of study in this paper, employers will gain experience with it and are able to learn what is the expected ability of students with a particular grade in a school. In this paper we will simplify notation by identifying the grade with the expected ability of the students with that grade.  In this sense, a school's strategy in this grading game is the way its students are grouped. In the two extreme cases grading can be completely informative, reveal the true abilities of each student, or can be totally non-informative, reveal only the average ability of all the students. Since we identified a grade with the average ability of the group of students with this grade, if a school announces grade $\hat{a}$ for a group of students, the true ability of a random student in that group is $\hat{a}$ in expectation. We call the grade announced for a student, the \textit{expected ability} of that student.

More formally, each school chooses a \textit{transcript structure}, which is a mapping from the abilities  into \textit{expected abilities $\hat{a} \in [0,1]$}. This mapping may be stochastic i.e, for each ability $a$ there can be a probability distribution over the set of expected abilities $\hat{a}$ that a student of ability $a$ can get. However, the average ability of students mapped to expected ability $\hat{a}$ in a given school is equal to $\hat{a}$.

In Section \ref{sec:general} we do not make any assumptions about the grading policies of schools. We rather consider the game of early contracting between students and jobs when played based on the given grading policies by the schools.  In Section \ref{sec:static} we assume that schools are fully strategic, aiming to maximize the average quality of jobs their students take. In this game the schools are the only strategic players, and after announcement of the grades, the matching between  student pool and job pool is one that maps higher expected abilities to higher quality jobs (as by Ostrovsky and Schwarz \cite{Ostrovsky_SchwarzAEJ10} fully strategic grading eliminates the incentive for early contracting).


\begin{definition}
Let $I$ be the number of school types and $\phi=(F_1,F_2,...,F_I)$ be a profile of transcript structures.\footnote{Recall that there are infinitely many schools in each type and each school can choose a different transcript structure. By $F_i$ we mean the average transcript structure over the schools of the same type. If the transcript structure chosen by each school is optimal, choosing the average transcript structure by all of them is also optimal \cite{Ostrovsky_SchwarzAEJ10}.} Function $Q(.)$ on $[0,1]$ is the desirability mapping corresponding to $\phi$, if the expected desirability of a position matched with a student labeled as expected ability $\hat{a}$ is equal to $Q(\hat{a})$ when schools use transcript structures in $\phi$.
\end{definition}

\paragraph{\bf Equilibrium}
Let  $\phi=(F_1,F_2,...,F_I)$ be a profile of transcript structures and consider the corresponding aggregate distribution of expected students abilities $G(.)$ and the resulting desirability mapping $Q(.)$.
Take any school \hbedit{from type} $i$, its transcript structure $F_i$ under $\phi$, and the resulting distribution
of expected students abilities in the school, $G_i(.)$. Consider any alternative transcript
structure $F'_i$ for this school. Notice  if only one infinitesimally small school changes its grading structure, this will not effect the  distribution of expected students abilities, and hence will not change the desirability mapping $Q(.)$. Let $\hat{G}_i(.)$ denote the distribution of expected students abilities in this school using the alternate grading $\hat{F}_i$.
We say that a profile $\phi$ of grading structures is in equilibrium if for any school type $i$
and any alternative transcript structure $\hat{F}_i$, the average student placement at the school
 under the original transcript structure is at least as high as it is under the alternative
one, keeping desirability mapping $Q(.)$ fixed:
$$\int_{0}^{1} Q(\hat{a})dG_i(\hat{a}) \geq
\int_{0}^{1} Q(\hat{a})d\hat{G}_i(\hat a)$$

\begin{definition}
Let $\hat{a}_L$ be the lowest and $\hat{a}_H$ the highest expected ability levels produced in an equilibrium. Then we say that the equilibrium is connected if for every point $\hat{a} \in (\hat{a}_L, \hat{a}_H)$ there exists a school that produces students of all expected abilities in some $\epsilon$-neighborhood of $\hat{a}$.
\end{definition}

\begin{definition}
An equilibrium is fully informative at a particular value of position desirability $q$ if there is an ability level that is necessary and sufficient for receiving a position of this quality.
\end{definition}

\paragraph{\bf The early contracting game}
In Sections \ref{sec:general} and \ref{sec:early} we consider the unraveling of matching market, a strategic game in which companies and students can improve their assignments by contracting early. To simplify the presentation, we will assume that there is only one time step in which early contracts can be offered or accepted:
We assume that companies offer jobs to students before any grade information is available from schools, simply based on the average abilities of the students in a given school. We assume that at this stage, neither the company nor the student has additional information about the student's abilities beyond the school of attendance.

We formalize this phenomenon as a two-stage game: Matching of jobs to students takes place in two phases, the first one is the early contracting phase where jobs can be offered to students in particular schools,  the second phase is the regular increasing assignment of jobs to students. In Section \ref{sec:early} we assume that schools reveal full information about the students in the second stage. In Section \ref{sec:general} the assignment after the early contracting stage will be based on the announced grading policy of the schools, but we do not assume that grading is fully informative or fully strategic. We assume that both employers and students are risk neutral, thus a student accepts an early contract if the quality of the job offered exceeds the expected quality he/she will receive in the second stage, and similarly an employer will want to offer a job to a random student in a school, if the expected ability of the student is higher than the expected ability of the student the job is assigned to in the second phase.

\paragraph{\bf An example of early contracting}
We start with an example to illustrate that early contracting can benefit both students and schools. Assume that there is only one type of school with uniform student distribution, and the job distribution is such that 3/4 of the jobs have quality distributed uniformly in the $[\frac{1}{2},1]$ range, while the other $1/4$ of the jobs have quality distributed uniformly in the $[0,\frac{1}{2}]$ range. Since students do not know their own ability, a student's expected job quality before attending school is $\frac{1}{4}\frac{1}{4}+\frac{3}{4}\frac{3}{4}=\frac{5}{8}$, so a student is happy to accept a job of quality $q>5/8$. On the other hand, consider a job with quality $q=\frac{5}{8}$. The fraction of jobs with higher quality is $9/16$, so $9/16$ of the students will take a better quality job, leaving the job with quality $q$ for a student with ability $a=7/16<1/2$. Consequently a job with quality $q'$ just above $q$ is ripe for early contracting. Its quality is better than the quality a student would expect to take, while the job would go to a student with less than average ability if it waits for transcripts.

\begin{definition}
\label{def:early-eq}
A set of early contracts followed by an assignment of students to jobs using a transcript structure is at equilibrium, if the following conditions hold using $\hat Q$ to denote the increasing mapping of the remaining students to the remaining jobs after early contracting.
\begin{enumerate}
\item if a student from a school with average ability $\hat a$, accepted a job with quality $q$ in early contracting, then
    \begin{enumerate}
    \item $\hat Q^{-1}(q)\le \hat a$ (else the job would prefer not to contract early),
    \item the average quality of the jobs $\hat Q$ assigned to remaining students in school, is at most $q$ (else the student prefers to wait and not contract early)
    \end{enumerate}
\item  there is no pair of schools and jobs, such that the average ability of the students in the school is $\hat a$ some of which are assigned only in the second stage, and the quality of job is $q$ with $\hat Q^{-1}(q)<\hat a$, and the average quality of the job students in the school get in the second phase is less than $q$ (as otherwise this pair of job and student would prefer to accept an early contract).
\item if jobs of quality $q<q'$ both are assigned in early contracting, the better job should go to a student from a school that is no worse in average ability (as otherwise the student from the better school, and the job of quality $q'$ would prefer to contract with each other).
\end{enumerate}
\end{definition}

Note that we made no assumption on the rationality of the transcript structures. In the special case of fully informative transcripts we face the special case of the classical unraveling game. With fully rational transcripts,
Ostrovsky and Schwarz \cite{Ostrovsky_SchwarzAEJ10} show  at the unique equilibrium of the grading game (Theorem \ref{thm:OS_unique}), there is no incentive to make early offers. We do not believe that schools are fully strategic in grading.
In Sections \ref{sec:general} and \ref{sec:early} we consider the effect of early contracting, when the schools do not optimally suppress information. Note that the early contracting game and  notion of its equilibrium we defined, assumes  the grading policy of each school is public information, and hence students and employers can know the quality of their matching if they wait for the second stage.

\section{General Setting}\label{sec:general}
Here we consider the general notion of early contracting equilibrium in Definition \ref{def:early-eq}, and find loss of efficiency at equilibrium mappings. The general results imply bounds on the loss of efficiency in both the strategic grading and the early contracting games (as they correspond to special grading structures used by schools).

We begin by proving that equilibria of this game result in increasing mappings of job to expected abilities. From the employer's perspective, the expected ability of a student matched in the first phase is the average ability of students of the school he/she attends, while the expected ability of a student matched in the second phase is the expected ability assigned to the student by the transcript structure of his/her school.

\begin{lemma}\label{increasing}
At the equilibrium of the general early contracting game, the mapping of jobs to expected abilities is non-decreasing.
\end{lemma}
\begin{proof}
We prove that the mapping is increasing by contradiction. Suppose $q_1<q_2$ are matched to $a_1>a_2$ respectively. The job $q_2$ is not at equilibrium, and would benefit by offering an early contract to $a_1$ or waiting till the second phase (depending on which phase $a_1$ is matched).
\end{proof}


Next we give our general bound for the price of anarchy. Recall that a student of ability $a$ assigned to a job of quality $q$ contributes $f(q)g(a)$ to social welfare.
\begin{theorem}\label{bound2}
Suppose functions $f$ and $g$ are increasing and $g$ is  concave.
Assuming overall student abilities are uniformly distributed, the efficiency loss of matching is at most a factor of 2.

\end{theorem}

\begin{proof}
First note that the students assigned via early contracting are random students from each school, therefore in each school the distribution of student abilities for the group of students with any grade $\hat{a}$ remains also $\hat{a}$  among students who didn't contract early.  Consider the expected value of assigning a job of quality $q$ to a student with expected ability $a$. Let $r$ be the random variable of the student's ability (so $E(r)=a$). Now the expected contribution to welfare is defined as $E(f(q)g(r))$. Since $g(.)$ is concave, $g(r)\geq rg(1)$ for $0\leq r \leq 1$. Therefore the expected value of $E(g(r))$ is more than $E(g(r))\ge E(rg(1))=E(r)g(1)=ag(1)$. The expected social welfare is greater than or equal to $\int_0^1 f(q)g(1)Q_{eq}^{-1}(q)d\mu(q)$. The mapping in equilibrium is increasing therefore, this is greater than $avg(G) \int_0^1 f(q)g(1)d\mu(q)$, where $avg(G)=\frac{1}{2}$ is average ability of students. Since
the optimum social welfare is clearly bounded by
$\int_0^1 f(q)g(1)d\mu(q)$, this proves the claimed bound of 2 on the loss of efficiency.
\end{proof}

We note that the proof of Theorem \ref{bound2} can be used without the assumption of overall uniform student distribution. The following corollary states the resulting bound, as well as corollaries in  special case when $f$ and $g$ are  identity functions. 

\begin{corollary}\label{cr:cr}
Suppose functions $f$ and $g$ are increasing and $g$ is  concave. The price of anarchy of strategic grading and early contracting games are bounded by  $\frac{1}{avg(G)}$. 

When student's ability is identified by his/her ability to contribute to social welfare (that is $g(a)$ is the identity function), then the social welfare at equilibrium is no worse than a random assignment. 

When value of an interaction is expressed as $aq$, the price of anarchy can also be bounded by  $\frac{1}{avg(\mu)}$.
\end{corollary}

\begin{proof}
The first statement follows directly from the proof above. To compare to a random assignment, note that in the random assignment, every job is mapped to average ability in expectation, and with $g(a)=a$, the resulting welfare is exactly $avg(G)\int_0^1 f(q)d\mu(q)$, the bound used in the proof.
With $f(q)=q$, the efficiency of random matching is the same as matching all the students to average job and is at least $\frac{1}{avg(\mu)}$ of the optimal assignment. Thus PoA is bounded by $\frac{1}{avg(\mu)}$.
\end{proof}

\section{Fully Strategic Grading Game}
\label{sec:static}

A bound of 2 for the price of anarchy for the strategic grading game follows from the results in Section \ref{sec:general}. In this section we assume functions $f$ and $g$ are both the identity, and show a tighter bound on the loss of efficiency of the matching process when schools act fully strategically and optimally suppress grade information to improve the placement of their students.  Ostrovsky and Schwarz \cite{Ostrovsky_SchwarzAEJ10} proved that if this game has any connected equilibrium, it has a unique one which only depends on aggregate student abilities and job qualities distributions:

\begin{theorem}\cite{Ostrovsky_SchwarzAEJ10}
\label{thm:OS_unique}
The aggregate distribution of expected abilities in any connected equilibrium is uniquely determined by the distribution of position desirabilities and the aggregate distribution of true abilities.
\end{theorem}

As extensively discussed in \cite{Ostrovsky_SchwarzAEJ10} connectedness is a mild restriction on the set of solutions. If there exists a school that gives out the lowest and highest grades and everything in between, the solution is connected. On the other hand non-connectedness means that there are no schools that give grades in an interval, while there are schools that give grades in values higher and lower than that interval, which sounds unreal.


We denote the connected equilibrium by $Q_{eq}$. Function $Q_{eq}$ maps expected abilities of students defined by the grading policies of the schools to job qualities as defined in Section \ref{sec:prelim}.

The main result of this section is the following bound on the price of anarchy.

\begin{theorem}
\label{thm:static}
If the students aggregate ability distribution is uniform, and $f$ and $g$ are the identity, the price of anarchy of the strategic grading game is bounded by $1.36$. 
\end{theorem}

We will focus for most of this section on the special case when the equilibrium mapping $Q_{eq}$ is linear. Then we use induction on the number of segments of $Q_{eq}$ to extend the bound to the general case.

The main technical lemma shows that jobs distribution that result in linear equilibrium have a higher average than uniform distribution at each point.

\begin{lemma}\label{lm:constraint}
If $Q_{eq}(.)$ is linear, when aggregate distribution of students is uniform, $Q_T$ satisfies the following condition for every $q' \in (0,1)$:
$$\int_0^{q'}q \mu'(q)dq \geq \frac{\mu(q')q'}{2}$$

\end{lemma}

The basic intuition is the following. When the average quality of the jobs below $q'$ is above $q'/2$, the schools have incentive to mix students. Mixing 0 quality students and $q'$ quality students improves their average placement. On the other hand, if at some point $q'$ the average quality of the jobs below $q'$ is less than $q'/2$, then the same mixing would hurt the average placement. The lemma shows that, when the average quality of  jobs below $q'$ is less than $q'/2$, the equilibrium $Q_{eq}$ contains a convex segment where the grading is fully informative. Linear equilibrium results from the mixing of students at all ability levels.

The proof of the lemma builds heavily on the equilibrium construction of Ostrovsky and Schwarz \cite{Ostrovsky_SchwarzAEJ10}. We will review the construction, and then prove the lemma in the Appendix.

Using this lemma we can prove our main theorem for the case when the equilibrium is linear. 

\begin{lemma}\label{lm:1.36}
If $Q_{eq}(.)$ is linear, when aggregate distribution of students is uniform, the price of anarchy is bounded by $1.36$.
\end{lemma}
\begin{proof}
For the case of linear equilibrium, we can express the price of anarchy as a ratio of integrals as
$$\dfrac{2\int_{0}^{1}qd\mu(q)\int_0^1 qQ_T^{-1}(q)d\mu(q) }{\int_0^1 q^2d\mu(q)}$$
We prove the lemma using the bound of Lemma \ref{lm:constraint}. See the Appendix for the details.
\end{proof}

\begin{proofof}{Theorem~\ref{thm:static}}
When the aggregate distribution of students is uniform, the equilibrium only depends on $\mu(q)$, the distribution of jobs. The equilibrium consists of strictly convex parts where $Q_{eq}(q)=Q_T(q)$, and linear parts where the endpoints of them are fully informative. Since $\frac{\sum A_i}{\sum B_i}$ is at most $\max_i \frac{A_i}{B_i}$, it is enough to prove that the ratio of efficiency of optimum to equilibrium in each part is less than $1.36$. The ratio in strictly convex intervals is $1$. And the ratio of the linear parts is bounded by 1.36 by Lemma ~\ref{lm:1.36}. 
\end{proofof}


\section{Early Contracting with Informative Grading}
\label{sec:early}
In this section we consider the early contracting game when schools grade fully informatively, revealing the ability of their students exactly.
A bound of 2 for the price of anarchy for early contracting game follows from the results in Section \ref{sec:general}. Example \ref{lowerbound} in the appendix, shows a setting with price of anarchy $\approx 1.22$ with the functions $f$ and $g$ both the identity.

In this section we will focus on the special case when $f$ and $g$ are the identity functions, and all schools have uniformly distributed student abilities (not only the aggregate ability of students is uniform).
First we study the structure of equilibria of this case, and show that the equilibrium is essentially unique, and then show that the price of anarchy can be bounded by $4/3$.

\vskip 0.2in
\noindent{ \bf Structure of the Equilibrium }\label{EqStruct}

\noindent We claim that the jobs that contract early in equilibrium form a centered interval of job qualities. 

\begin{lemma}\label{lm:center}
The jobs which have contracted early in equilibrium form an interval of job qualities. Additionally, the interval of student abilities that this set of jobs would have gotten in the optimal assignment is symmetric around the line $x = \frac{1}{2}$.
\end{lemma}

\begin{proof} It is not hard to see that jobs contracting early must form a continuous interval at equilibrium, as otherwise some jobs or students did not act strategically. See Figure \ref{fig:2stage} for the structure of the optimal assignment and the assignment in the second stage of the early contracting game.

To show that the interval is symmetric around $1/2$, suppose the interval of students who would have been assigned to this set of jobs after seeing the transcripts, begins at $a$. We prove that it ends at exactly $1-a$. By removing the jobs in the interval, two sets of job remain. The jobs with quality higher than the jobs in the interval, which we call them \emph{"good jobs"} and the jobs with quality below the jobs in the interval, \emph{"bad jobs"}, with the average quality of all jobs lying between the good and bad jobs.

Suppose the interval did not end at $1-a$. If it had ended before $1-a$, it means that some of the good jobs are mapped with less than average ability students. In this case they would have been better off, had they offered in the first stage. If it had ended after $1-a$, some of the bad jobs are mapped with better than average students. Which means that in the first stage students should not have accepted the lowest jobs that they did.
\end{proof}




\begin{figure}
\centering
\includegraphics[height=4cm]{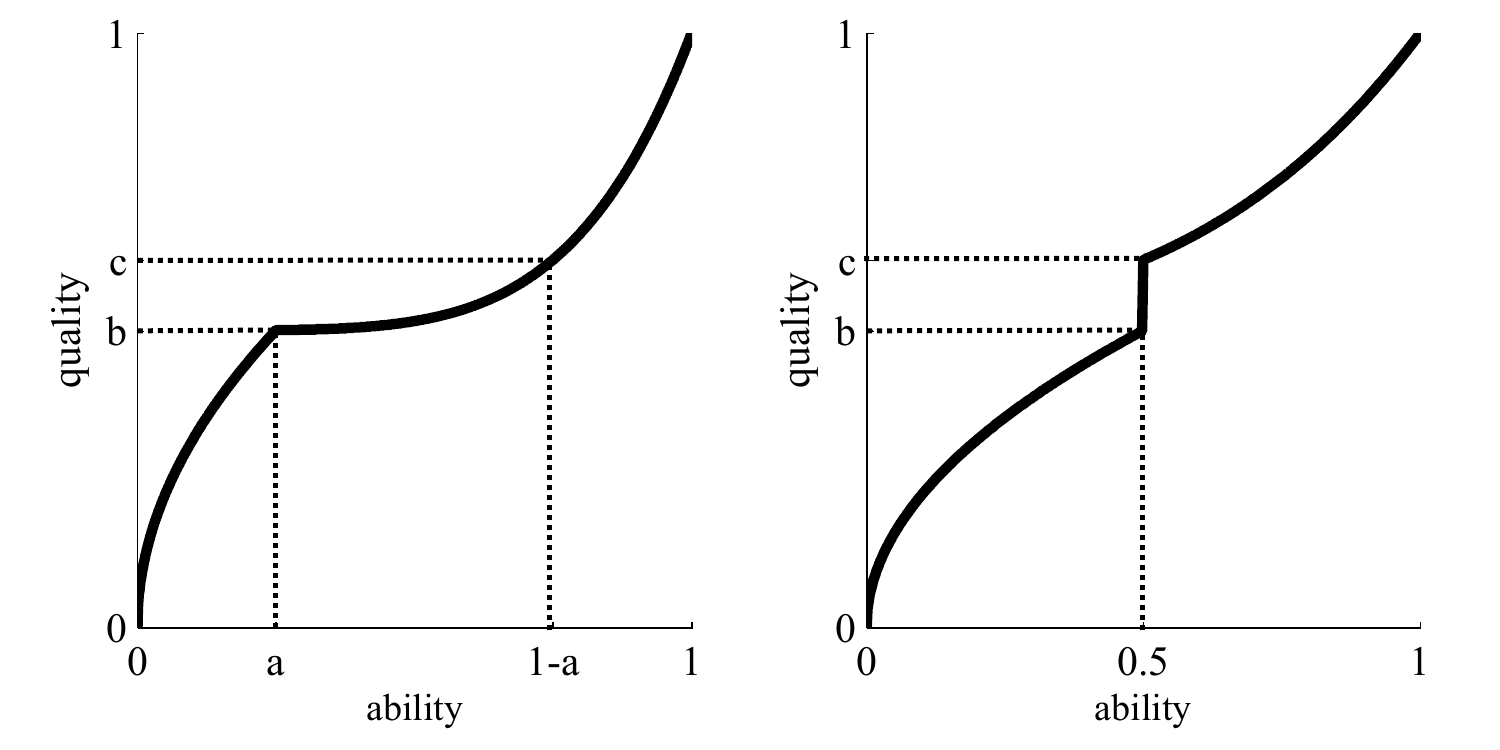}
\caption{The optimum mapping on the left. The mapping in phase 2 on the right. 
The interval of abilities corresponding to $[0,b]$ and $[c,1]$ now expand to now expands to $[0,0.5]$ and $[0.5,1]$ respectively.}
\label{fig:2stage}
\end{figure}


\begin{theorem}\label{uniformSchools}
The price of anarchy for the two-stage assignment game in the case that all of the schools have uniform distribution of students on the same interval of abilities is at most $\frac{4}{3}$.
\end{theorem}

\begin{proof}
Let $b$ and $c$ be the lowest and highest quality jobs respectively that participate in early contracting. Observe that we must have $c\le 2b$, as the quality of every job in the center interval should be at least the average quality of the jobs left, as otherwise students will not accept the early contract offer. The average quality of the remaining jobs is at least $c/2$ as half the jobs have quality at least $c$. The equilibrium mapping stretches mapping for low and high quality students, as shown on the right of Figure \ref{fig:2stage}. In the Appendix we show that the loss of efficiency due to this stretching can be bounded by 4/3.
\end{proof}

%


\bibliographystyle{abbrvnat}
\bibliography{ref}

\appendix
\label{sec:appendix}


\section{Appendix: Omitted Proofs}

In this section we prove the main technical lemma \ref{lm:constraint}, lemma \ref{lm:1.36} and theorem \ref{uniformSchools}. We conclude this section with lower bound examples for price of anarchy of various settings studied in the paper.

In order to prove lemma \ref{lm:constraint}, we need to recall some properties about the connected equilibrium shown by Ostrovsky and Schwarz \cite{Ostrovsky_SchwarzAEJ10}, as well as the process developed by  Ostrovsky and Schwarz \cite{Ostrovsky_SchwarzAEJ10} to find this equilibrium.

\begin{property}
The desirability mapping in equilibrium is an invertible, monotonically increasing, continuous function, i.e., no positive mass of students receives the same expected ability.
\end{property}

\begin{property}
The equilibrium mapping is a convex function of expected abilities.
\end{property}

\begin{property}\label{pr:al}
The lowest expected student ability in equilibrium, $\hat{a}_L$, is equal to the lowest true ability $a_L$, i.e. $0$.
\end{property}

\begin{property}\label{pr:strictlyConvex}
The equilibrium is fully informative at strictly convex points i.e. if $Q_{eq}$ is strictly convex at $q$, $Q_T^{-1}(q)=Q_{eq}^{-1}(q)$.
\end{property}

The last property also means that any student with expected ability $\hat{a} < a$ ($\hat{a} > a$) such that $(a,q)$ is a strictly convex point in the equilibrium has true ability less than (greater than) $a$. Otherwise the school can benefit by revealing the true grade for these students.

\begin{property}\label{pr:2ends}
Let the two ends of any linear section of $Q_{eq}$ be $(a_1,q_1)$ and $(a_2,q_2)$. Since the equilibrium is convex, $(a_1,q_1)$ and $(a_2,q_2)$ are either strictly convex (fully informative due to Property \ref{pr:strictlyConvex}) or the beginning or end points of $Q_{eq}$ curve. In either cases, the set of students assigned to job desirabilities $[q_1,q_2]$ in $Q_T$ and $Q_{eq}$ are the same. Which means the students with expected abilities in $[a_1,a_2]$ are the students whose true abilities are in the same interval and vice versa. Based on the rules for choosing transcript structures, the sum of expected abilities in this range is equal to the sum of abilities.
\end{property}

\begin{property}\label{pr:slope}
The slope of $Q_{eq}$ at any fully informative point which is start of a linear part is less than or equal to right derivative of $Q_T$ at that point.
\end{property}
\begin{proof}
Suppose the slope of the equilibrium is greater than right derivative of $Q_T$ at that point. There exists $\epsilon>0$ such that jobs in $[0, \epsilon]$ are mapped to strictly lower expected abilities than in $Q_T$. However in $Q_T$, lowest ability students take $[0, \epsilon]$ jobs and no bundling of students can introduce lower expected abilities.
\end{proof}

Next we need to recall the procedure in \cite{Ostrovsky_SchwarzAEJ10} for finding $Q_{eq}$ based on $Q_T$. This procedure is graphical and is inductive on  the number of convex/concave intervals. (We assume that $Q_T$ does not switch from convex to concave and vice versa infinitely many times.) The following finds the first segment of the equilibrium corresponding to the first convex/concave interval, the next segments can be found inductively.

Case 1: $Q_T$ starts with a concave part.




The equilibrium corresponding to first concave part is a line segment. This line can be found by the following procedure:

Consider the graph of $Q_T$ on a two-dimensional plane, and take the infinite ray that starts at point $(0,0)$ and has a slope of zero. Start rotating this ray around its origin, increasing its slope. Once the ray begins to intersect with the graph of $Q_T$ at points $(a_i,q_i)$ other than the origin, for each of these points keep checking whether they can be the end of first linear segment of equilibrium. Using property \ref{pr:2ends}, a necessary condition is that the sum of abilities assigned to $[0,q_i]$ desirabilities is equal to the sum of expected abilities when 
the equilibrium consists of linear segment connecting $(0,0)$ and $(a_i,q_i)$.

The sum of abilities of students in $[0,a_i]$ is $\frac{a_i^2}{2}$ since abilities are uniform, while the sum of expected abilities in $[0,a_i]$ corresponding to $Q_{eq}(\hat{a})=\frac{q_i}{a_i}\hat{a}$ is $\int_0^{q_i}Q_{eq}^{-1}(q)$ which is $\int_0^{q_i}\frac{a_i}{q_i}qd\mu(q)$, where $\mu(.)$ is the aggregate distribution of jobs qualities.

Continue until the slope is $\frac{q_H}{a_H}$, or in our case $45^\circ$. If there was a point with this property choose one that is on the lowest slope ray. And if there are multiple qualified points on one ray, choose one with maximum coordinates. Denote the point by $(a',q')$. $Q_{eq}$, starts with a linear section from $(0,0)$ to $(a',q')$ and is strictly convex in $(a',q')$. To find the equilibrium curve after this point, continue the procedure with remaining students and jobs.

In the case that we cannot find any such points with coordinates less than $(a_H,q_H)$, the equilibrium is a single linear segment. By property \ref{pr:al}, we know the lower expected ability in the equilibrium is equal to true lower ability; $\hat{a}_L=0$. After announcing the grades, the lower ability student is mapped to the lower quality job and the higher ability student is mapped to the higher quality job: $Q_{eq}(0)=0$, $Q_{eq}(\hat{a}_H)=1$. So $Q_{eq}(\hat{a})=\frac{\hat{a}}{\hat{a}_H}$. Finding $\hat{a}_H$ leads to exact formula for $Q_{eq}$. By applying property \ref{pr:2ends} about equality of sum of abilities and expected abilities, $\hat{a}_H$ is derived uniquely. Similar to the argument above for checking the validity of linear segment, $\hat{a}_H$ should be chosen in a way that the sum of expected abilities equals the sum of abilities. $Q^{-1}_{eq}(q)=q\hat{a}_H$, so the expected ability mapped to quality $q$ is $\hat{a}_Hq$. Thus the sum of expected abilities is $\int_0^1\hat{a}_Hqd\mu(q)$. The sum of abilities is equal to $0.5$ due to uniformity of students distribution. So $\hat{a}_H=\frac{0.5}{\int_0^1qd\mu(q)}$.

A special case of this $Q_T$ starting with a concave part is when $Q_T$ is concave. Concavity and the fact that $Q_T(0)=0$ and $Q_T(1)=1$ result in $Q_T$ being above rays with slope less than $45^\circ$. In this case the sweeping ray with slope less than or equal to $45^\circ$ does not cross $Q_T$ at a point with smaller coordinates than $(1,1)$.

Case 2: $Q_T$ starts with a convex part.



The equilibrium corresponding to first convex part lies on $Q_T$ up to some point and then continues with a linear segment. In extreme cases, this segment of equilibrium is fully informative or linear on the whole interval. By following the procedure below we construct this first segment.

This case is similar to previous one. The difference is that using property \ref{pr:slope} we make sure the slope of the ray does not exceed the right derivative of its origin. So other than increasing the slope we move the origin of the ray along $Q_T$, keeping the ray tangent to $Q_T$, to avoid exceeding the right derivative of the origin.
When the origin of the ray is at $(a',q')$, we know that $Q_{eq}$ coincides with $Q_T$ up to this point. Like the procedure for case 1, we check the intersection points $(a_i,q_i)$ in the intersection of the ray and $Q_T$ to see whether the line segment connecting $(a',q')$ and $(a_i,q_i)$ satisfies property \ref{pr:strictlyConvex}.

If origin of the ray reaches the last point of the first convex interval or if the slope of the ray exceeds the slope of last point of this interval, $Q_{eq}$ is fully informative on this interval.

A special case is when $Q_T$ is a single convex interval. In this case, $Q_{eq}$ is equal to $Q_T$ at every points.





\begin{lemma}\label{lm:45}
If $Q_{eq}$ corresponding to $Q_T$ is linear, $Q_T$ starts with slope $\geq \frac{q_H-q_L}{a_H-a_L}$. Thus in the special case where abilities are in the range $[0,1]$ and qualities are $[0,1]$, $Q_T$ starts with slope $\geq 45^\circ$.
\end{lemma}
\begin{proof}
For notational simplicity we prove the lemma for the special case mentioned in the lemma. Since the average of students abilities in a bundle cannot exceed maximum ability, the maximum expected ability is not greater than expected ability, i.e. $\hat{a}_H\leq 1$. Since $Q_{eq}(0)=0$ and $Q_{eq}(\hat{a}_H)=1$ and $Q_{eq}$ is linear, $Q_{eq}(\hat{a})=\frac{a}{\hat{a}_H}$. So the slope of equilibrium is $\geq 45^\circ$. Property \ref{pr:slope} implies that the slope of $Q_T$ at $(0,0)$ is not less than slope of $Q_{eq}$. Thus $Q_T$ starts with slope $\geq 45^\circ$.
\end{proof}

\begin{property}\label{ineq}
If $Q_{eq}$ corresponding to $Q_T$ is linear, for every point $(a',q')$ on $Q_T$, below line $Q(\hat{a})=\hat{a}$,
$$\int_0^{q'}\frac{q}{q'}a'\mu'(q)dq \neq \int_0^{a'}adG(a)$$
Therefore since $Q_T$ is continuous, every point on a continuous segment of $Q_T$ below $Q(\hat{a})=\hat{a}$ has the same direction in inequality.
\end{property}

This property is due to the constructive method for $Q_{eq}$. The equilibrium is linear when there does not exist such a point on a ray with origin $(0,0)$ and slope below $45^\circ$.

Now we are ready to prove the main lemma.

\begin{proofof}[ Lemma \ref{lm:constraint}]


We inductively prove the lemma for segments of $Q_T$ that are below or above the line $Q(\hat{a})=\hat{a}$. Assuming that the constraint holds for $Q_T$ points up to some intersection of $Q(\hat{a})=\hat{a}$ and $Q_T$ we prove it holds for the points till the next intersection.

Since $Q(\hat{a})=\hat{a}$ is the bisector of first quadrant, we call this line the bisector from now on.


We first prove that $Q_T$ curve visits the above side of the bisector before visiting the below part, or more formally for any point $(a,q)$ on $Q_T$ such that $q<a$, there exist $(a',q')$, such that $q'>a'$.

By lemma \ref{lm:45} we know that $Q_T$ either starts on the above part or coincides with the bisector up to some point. For slope $>45^\circ$, the curve starts in the above part and claim is obviously correct. Suppose $Q_T$ starts with slope $=45^\circ$ and visits below of the bisector first. Let $(a_1,a_1)$ be the last point before on the bisector before visiting the below part. Let $(a_2,a_2)$, the first point it crosses the bisector after visiting the below section. Note that there exists such a point because the curve eventually crosses the bisector $Q_T(1)=1$. Consider the segment between $(a_1,a_1)$ and $(a_2,a_2)$ on $Q_T$. By definition this segment is fully below the bisector. Consider the first time that the sweeping ray intersects with this segment and call the intersection point $(a',q)'$.
We claim that inequality in property \ref{ineq}, have different directions for $(a',q')$ and $(a_2,a_2)$ contradicting with property \ref{ineq}. The reason is that the line segment $(0,0),(a',q')$ is fully below $Q_T$ up to $(a',q')$ where the line segment $(0,0),(a_2,a_2)$ is fully above $Q_T$ up to $(a_2,a_2)$.

Using the argument above, we know that $Q_T$ satisfies one of the three cases below:
\begin{enumerate}
\item $Q_T$ is on the bisector: $Q_T(a)=a$ for $0\leq x \leq 1$
\item $Q_T$ goes to the above side of bisector at $b$ for the first time: $Q_T(a)=a$ for $0 \leq x \leq b$, $b>0$ the right derivative at $b$ is $> 45^\circ$
\item  $Q_T$ goes to the above side of the bisector at $0$: Right derivative at $0$ is $>45^\circ$
\end{enumerate}

Case 1 happens with uniform distribution of jobs. Thus the lemma constraint satisfies for all points.

Case 2 satisfies the constraint up to point $b$, because the job distribution for $0 \geq b$ is the same as uniform distribution.

The constraint holds for point $(0,0)$ in the third case and $(b,b)$ in the second case.

Suppose the constraint holds up to a point $a_i$, where $Q_T(a_i)=a_i$.
In case 1, below, we prove that this constraint holds with strict till the next crossing point.

Case 1: $Q_T$ curve goes to the above side of the bisector:

Consider the segments of the curve that lie on the above of the bisector. Each segment consists of an interval of abilities $[a_i,a_j]$ which is mapped to an interval of qualities $[q_i,q_j]$. Whether $q_i$ is $0$ or a point where the curve goes to above:
$$\dfrac{\int_0^{q_i}q \mu'(q)dq}{q_i} \geq \frac{\mu(q_i)}{2}$$
For any point $(a',q')$ on the curve such that $q' \in [q_i,q_j]$:
$$\dfrac{\int_0^{q'}q d\mu(q)}{q'}=
\dfrac{\int_0^{q_i}q d\mu(q)+\int_{q_i}^{q'}q d\mu(q)}{q'} \geq
\dfrac{\dfrac{q_i \mu'(q_i)}{2}+\int_{q_i}^{q'}q d\mu(q)}{q'}$$
By rewriting the integral we get:
$$\int_{q_i}^{q'}qd\mu(q)
=\int_{\mu(q_i)=a_i}^{\mu(q')=a'}\mu^{-1}(\mu(q))d\mu(q)$$
Change of variables results in:
$$=\int_{a_i}^{a'}\mu^{-1}(a)da$$
Since the distribution of abilities is uniform, for every $q$, $\mu(q)=Q_T^{-1}(a)$ which means $Q_T$ and $\mu$ are inverse of each other. Also $(a_i,q_i)$ are on the bisector, so $\mu(q_i)=a_i=q_i$ and $a'>q'$, therefore:
$$\int_{a_i}^{a'}\mu^{-1}(a)da > \int_{q_i}^{q'}\mu^{-1}(a)da \geq \int_{q_i}^{q'} a da $$

Where the last inequality holds since $\mu^{-1}(a)$ is above $a$. So the average of qualities up to $q'$ is:
$$\dfrac{\int_0^{q_i}q \mu'(q)dq}{q_i} \geq
\dfrac{\dfrac{q_i \mu(q_i)}{2}+\dfrac{q'^2}{2}-\dfrac{q_i^2}{2}}{q'}$$
$(a_i,q_i)$ is on the bisector and $\mu(q_i)=q_i$, therefore for every $q'$:
$$\dfrac{\int_0^{q'}q d\mu(q)}{q'} \geq \frac{\mu(q')}{2}$$

Now we prove that if the constraint holds with strict inequality for point $a_i$, where $Q_T(a_i)=a_i$, it will hold till the next crossing point

Case 2: $Q_T$ curve goes to the below side of the bisector:



Consider the inequality in property \ref{ineq}. For $(a,q)$ on the bisector, this inequality matches this lemma's constraint. Thus by induction assumption the left hand side is greater than the right hand side. Property \ref{ineq} implies that the inequality holds the same direction for all points on the continuous segment below the bisector.

Thus for every point $(a',q')$ on this interval:
%
$$\int_0^{q'}\frac{q}{q'}a' \mu'(q)dq > \int_0^{a'}adG(a)=\int_0^{a'}ada=\frac{a^{\prime^2}}{2}$$
 By dropping $a'$ from both sides and moving $q'$ to the other side, we reach:
$$\int_0^{q'}q \mu'(q)dq >\frac{a'q'}{2} = \frac{\mu(q')q'}{2}$$
Where the last equality holds due to property \ref{pr:reverse}. This proves the claim.

\end{proofof}

\begin{proofof}{lemma \ref{lm:1.36}}


 The price of anarchy is equal to:
$$\dfrac{\int_{q_L}^{q_H}qQ_T^{-1}(q)d\mu(q)}{\int_{q_L}^{q_H}qQ_{eq}^{-1}(q)d\mu(q)}$$
In the equilibrium, the sum of expected abilities is equal to the sum of real abilities. So when the equilibrium is linear the maximum expected ability, $\hat{a}_H$, is chosen such that:
$$\int_{0}^{1}q\hat{a}_Hd\mu(q)=\int_0^1adG(a)$$
Since students distribution is uniform the above value is $\frac{1}{2}$, which means:
$$\hat{a}_H=\dfrac{1}{2\int_{0}^{1}qd\mu(q)}$$
Since in the equilibrium, quality $q$ is mapped to $q\hat{a}_H$, the efficiency of the equilibrium is:
$$\int_0^1 q\cdot q\cdot \hat{a}_H = \dfrac{1}{2\int_{0}^{1}qd\mu(q)}\int_0^1 q^2d\mu(q) $$
So PoA is as follows:
$$\dfrac{2\int_{0}^{1}qd\mu(q)\int_0^1 qQ_T^{-1}(q)d\mu(q) }{\int_0^1 q^2d\mu(q)}$$

To find an upper bound on PoA, we use the constraint in Lemma \ref{lm:constraint}. We can find the highest value for $\int_0^1 qQ_T^{-1}(q)d\mu(q)$ and lowest value for $\int_0^1 q^2d\mu(q)$, for possible values of $\int_0^1 qd\mu(q)$.

Average of qualities is equal to $\int_0^1 qd\mu(q)$. Lemma \ref{lm:constraint} states that for every value of job quality, $q'$, average of job qualities that are $\leq q'$, is better than the average of uniform job distribution up to that point.

Lemma \ref{lm:constraint} implies that average of jobs is above $0.5$. Suppose that $\int_0^1 qd\mu(q)=a$. We claim that the minimum value for $\int_0^1 q^2d\mu(q)$ is $a^2$.

Cauchy-Schwarz inequality states that:
$$\left(\int_0^1 1. qd\mu(q)\right)^2 \leq \int_0^1 1^2d\mu(q)  \int_0^1 q^2d\mu(q)$$
$$a^2 \leq \int_0^1 q^2d\mu(q)$$

We now prove an upper bound for $\int_0^1 qQ_T^{-1}(q)d\mu(q)$. We claim that the maximum for this integral, while satisfying Lemma \ref{lm:constraint}, occurs when $x$-fraction of jobs are uniformly distributed in $[0,x]$ and others are of quality $1$. The value of $x$ is determined by average of job qualities.

First we show that if there was not such a constraint, the distribution with maximum $\int_0^1 qQ_T^{-1}(q)d\mu(q)$ was such that jobs are either of $0$-quality or $1$-quality.

The integral $\int_0^1 qQ_T^{-1}(q)d\mu(q)$ can be rewritten as an integral over abilities axis:
$\int_0^1 aQ_T(a)da$, which is a weighted sum of abilities. Since the sum of the weights, $\int_0^1 qd\mu(q)$, is fixed, the maximum weighted sum occurs when the high abilities get the highest possible weight, $1$, and the low abilities get the lowest possible weight, $0$.

To apply the constraint in Lemma \ref{lm:constraint}, the lowest possible weights that low abilities can get is no longer $0$ but it is the weights resulted from a uniform distribution, which is weight equal to that ability value.

Now, for $\int_0^1 qd\mu(q)=a$ we find $x$ such that if $x$ fractions of the jobs are uniformly distributed from $0$ to $x$ and others are $1$ this results in the required total weight of $a$:
$$\int_0^xqdq+\int_x^1 1 dq=a$$
Lemma \ref{lm:constraint} implies that $\int_0^1 qd\mu(q) \in [\frac{1}{2},1]$. For $\frac{1}{2} \leq a \leq 1$, $x=1-\sqrt{2a-1}$.
Therefore
$$\int_0^1 a'Q_T(a')da'=\int_0^x a'^2da'+\int_x^1 a'\cdot 1da'$$
$$=\dfrac{x^3}{3}+\dfrac{1-x^2}{2}=\dfrac{(1-\sqrt{2a-1})^3}{3}+\dfrac{1-(1-\sqrt{2a-1})^2}{2}$$

$$PoA \leq \max_a \dfrac{2a\left[ \dfrac{(1-\sqrt{2a-1})^3}{3}+\dfrac{1-(1-\sqrt{2a-1})^2}{2} \right]}{a^2}$$

$$< 1.36$$
\end{proofof}




\begin{proofof}{theorem \ref{uniformSchools}}

The main proof idea is as follows. We split the job qualities region $[0,1]$ into 3 parts using the structure of the equilibrium we described above. The first part $[0,b]$, the second part $[b,c]$, and the third $[c,1]$, where the interval of jobs $[b,c]$ are the jobs in early contract, corresponding to students of quality $a$ to $1-a$ in the optimal assignment (as shown in the equilibrium description and figure \ref{fig:2stage}). The truthful mapping, $Q_T(.)$, is the function that maps abilities to qualities, had there been no early contracting. The equilibrium mapping, $Q_{eq}(.)$, is an equilibrium mapping, when early contracting is allowed. We abuse the notation and use it for both phases. However since we proved the jobs in interval $b,c$ are the jobs who participate in early contracting (the first phase), it is clear that, by $Q_{eq}^{-1}(q)$ we mean the first phase when $q \in [b,c]$ and second phase otherwise.\\
$$
\text{Price of Anarchy (PoA)} = \dfrac{\int_0^1 qQ_T^{-1}(q)\mu'(q)dq}{\int_0^1 qQ_{eq}^{-1}(q)\mu'(q)dq}
$$

$$
=\dfrac{\int_0^b qQ_T^{-1}(q)\mu'(q)dq+\int_b^c qQ_T^{-1}(q)\mu'(q)dq+\int_c^1 qQ_T^{-1}(q)\mu'(q)dq}{\int_0^b qQ_{eq}^{-1}(q)\mu'(q)dq+\int_b^c qQ_{eq}^{-1}(q)\mu'(q)dq+\int_c^1 qQ_{eq}^{-1}(q)\mu'(q)dq} \\
$$

Then, we show that the ratio of the social welfare in the optimum to the equilibrium  in each of these regions is at most $\frac{4}{3}$.

\begin{enumerate}
\item Consider the lower interval $[0,b]$

$Q_{eq}$ is an expansion of $Q_T$, such that its domain is $[0,\frac{1}{2}]$ instead of $[0,a]$ and the same job is mapped to a higher ability person. Thus for every point in this interval, $Q_T^{-1}(q)$ is not larger than $Q_T^{-1}(q)$ and the ratio of $\int_0^b qQ_T^{-1}(q)\mu'(q)dq$ to $\int_0^b qQ_{eq}^{-1}(q)\mu'(q)dq$ is less than $1$ and therefore $\frac{4}{3}$.

\item Consider the middle interval: $[b,c]$.


In the first phase, the jobs are assigned to average student (ability = $\frac{1}{2}$), The efficiency of the equilibrium is:

$$\int_b^c qQ_{eq}^{-1}(q)d\mu(q) = \frac{1}{2}\int_b^c Q_{eq}^{-1}(q)d\mu(q)$$

The set of jobs in $[b,c]$ is the same in the equilibrium and optimum:
$$\int_b^c Q_{eq}^{-1}(q)d\mu(q)=\int_b^c Q_T^{-1}(q)d\mu(q)$$

By changing the axis of integration to abilities we have:

$$\dfrac{\int_b^c qQ_T^{-1}(q)d\mu(q)}{\int_b^c qQ_{eq}^{-1}(q)d\mu(q)}
=  \dfrac{\int_a^{1-a} a'Q_T(a')da'}{\int_a^{1-a} \frac{1}{2}Q_T(a')da'}$$




To find an upper bound for this fraction we use Claim \ref{claim} below, that finds that the maximum value of the numerator with a fixed denominator occurs  when the fraction $[1-a,x]$ of students get jobs with quality $b$ and $[x,a]$ get jobs of quality $c$, where $x$ is chosen such that $b(x-a)+c(1-a-x)=\int_a^{1-a} Q_T(a')da'$. Using this fact we get

$$
\leq \dfrac{\int_a^xa'b da'+ \int_x^{1-a}a'c da'}{\frac{1}{2}\left(\int_a^xb da' + \int_x^{1-a}c da' \right) }
= \dfrac{\frac{a+x}{2}(x-a)b+\frac{1-a+x}{2}(1-a-x)c}{\frac{1}{2}\left[(x-a)b + (1-a-x)c\right]}
$$

The above fraction is decreasing in $a$ and increasing in $c$. To get an upper bound we replace $a$ by $0$. Also recall from proof of theorem \ref{uniformSchools} that $c\leq 2b$:

$$
\leq \dfrac{\frac{x^2}{2}+(1-x^2)}{\frac{1}{2}(x+2(1-x))}=\dfrac{2-x^2}{2-x} < \dfrac{4}{3} \text{  (for all $0\leq x \leq 1$)}
$$

\item Interval $[c,1]$.


The social welfare of the optimum in this interval can be rewritten over abilities axis as below:
$$\int_c^1 qQ_T^{-1}(q)d\mu(q)=\int_{1-a}^1 a'Q_T(a')da'$$

Equilibrium function, $Q_{eq}$, is an expansion of $Q_T$ from $[1-a,1]$ to $[\frac{1}{2},1]$, such that
$(1-(1-a'))\frac{0.5}{a}$ gets a job quality that $a'$ gets in the optimum, so:

$$\int_c^1 qQ_{eq}^{-1}(q)\mu'(q)dq=\int_{1-a}^1(1-(1-a')\frac{0.5}{a})Q_T(a')da'.
$$
So, our goal is to upper bound:
$$\dfrac{\int_{1-a}^1 a'Q_T(a')da'}{\int_{1-a}^1(1-(1-a')\frac{0.5}{a})Q_T(a')da'}$$

The ratio $\frac{a'}{1-(1-a')\frac{0.5}{a}}$ ranges from $1$ to $\frac{1-a}{0.5}=2(1-a)$ as $a'$ ranges from $1$ to $1-a$ and is monotone decreasing in this range, and $Q_T(a')$ is increasing. In other words, $Q_T(.)$ assigns more weight to lower ratios.


By replacing $Q_T(a')$ with a constant value for every $a'$, we allocate equal weight to all the ratios $\frac{a'}{1-(1-a')\frac{0.5}{a}}$:

$$
\leq \dfrac{\int_{1-a}^{1} a'da'}{\int_{1-a}^{1} (1-(1-a')\frac{0.5}{a})da'}
= \dfrac{1-\frac{a}{2}}{\frac{3}{4}} \leq \dfrac{4}{3}
$$

where the last inequality follows as $0\leq a \leq \frac{1}{2}$.
\end{enumerate}
\end{proofof}

\begin{claim}
\label{claim}
For job qualities $\in [b,c]$ and a fixed value for sum of qualities, $\int_a^{1-a} Q_T(a')da'$, the maximum social welfare, $\int_a^{1-a} a'Q_T(a')da'$, occurs when the fraction $[1-a,x]$ get jobs with quality $b$ and $[x,a]$ get jobs of quality $c$, where $x$ is chosen such that $b(x-a)+c(1-a-x)=\int_a^{1-a} Q_T(a')da'$.
\end{claim}
\begin{proof}
The integral $\int_a^{1-a} a'Q_T(a')da'$ is a weighted sum over abilities, where the job quality ($Q_T(a')$) mapped to ability ($a'$), is the weight of that ability in the summation.
Since $Q_T(a') \in [b,c]$, every $a' \in [1-a,a]$ gets weight at least $b$ and at most $c$. Suppose we first assign weight $b$ to every  $a' \in [1-a,a]$. Now we need to increase the assigned weights to reach $\int_a^{1-a} Q_T(a')da'$. Increasing a higher ability weight results in a larger increase in $\int_a^{1-a} a'Q_T(a')da'$ than a lower ability. So we start from $1-a$ and continue to lower abilities and increase the weight as much as possible. Since maximum weight is $c$, the final assignment is $c$ to fraction $[x,1-a]$ and $b$ to $[a,x]$.
\end{proof}

\section{ Lower bound examples for the price of anarchy}

\noindent The following example finds a lower bound on the price of anarchy for fully strategic schools game.
In this example, the ratio between the efficiency of optimum and equilibrium mapping is approximately $1.07$, this means that price of anarchy is $\geq 1.07$.

\begin{example}\label{ex:1.07}
The family of functions $Q_T(a)=a^{\frac{1}{x}}$ with $1 \leq x$, form concave curves and therefore linear equilibria. The procedure for finding equilibrium, gives the formula for the special case that $Q_T$ is concave. The equilibrium corresponding to $Q_T(a)=a^{\frac{1}{x}}$ is $Q_{eq}(a)=\frac{2x}{x+1}a$. Property \ref{pr:reverse} implies that jobs distribution, $\mu$ is inverse of $Q_T$, i.e. $\mu(q)=q^x$.

The ratio between optimal efficiency and equilibrium efficiency is $\frac{\int_0^1qQ_T^{-1}(q)d\mu(q)}{\int_0^1qQ_{eq}^{-1}d\mu(q)}$. By plugging in the formulas, we have $\frac{2x(x+2)}{(x+1)(2x+1)}$. The maximum happens at $\sim 2.73$ with value $\sim 1.07$.

\end{example}

The following example finds a lower bound on the price of anarchy for early contracting game with multiple types of schools. The distribution of students among schools and jobs are such that the ratio of optimum mapping to the equilibrium is $\frac{11}{9} \approx 1.22$.

\begin{example}\label{lowerbound}
Suppose there are two types of schools, $A$ and $B$. Student abilities in schools of type $A$ are uniformly distributed in intervals $[0,\frac{1}{4}-\epsilon]$ and $[\frac{3}{4},1]$ where $\epsilon$ is a positive value $\approx 0$. Type $B$ students however, have abilities uniformly distributed on $[\frac{1}{4}-\epsilon,\frac{3}{4}]$. Also suppose that $\frac{1}{4}$ of jobs are of quality $0$, $\frac{1}{2}$ are of quality $\frac{1}{2}+\epsilon$ and other $\frac{1}{4}$ are of quality $1$.\footnote{To be compatible with distribution constraints in section \ref{sec:prelim}, $\mu$ can be defined this way: $\frac{1}{4}-\epsilon$ fraction of jobs are distributed uniformly on $[0,\epsilon]$, $\frac{1}{2}-\epsilon$ fraction on $[\frac{1}{2},\frac{1}{2}+\epsilon]$ and $\frac{1}{4}-\epsilon$ fraction on $[1-\epsilon,1]$, and $3\epsilon$ fraction on $[\epsilon,\frac{1}{2}\ and [\frac{1}{2}+\epsilon,1-\epsilon]$.  } The truthful mapping, $Q_T$, is illustrated in figure \ref{fig:lb}.
\begin{figure}
\centering
\includegraphics[height=6.2cm]{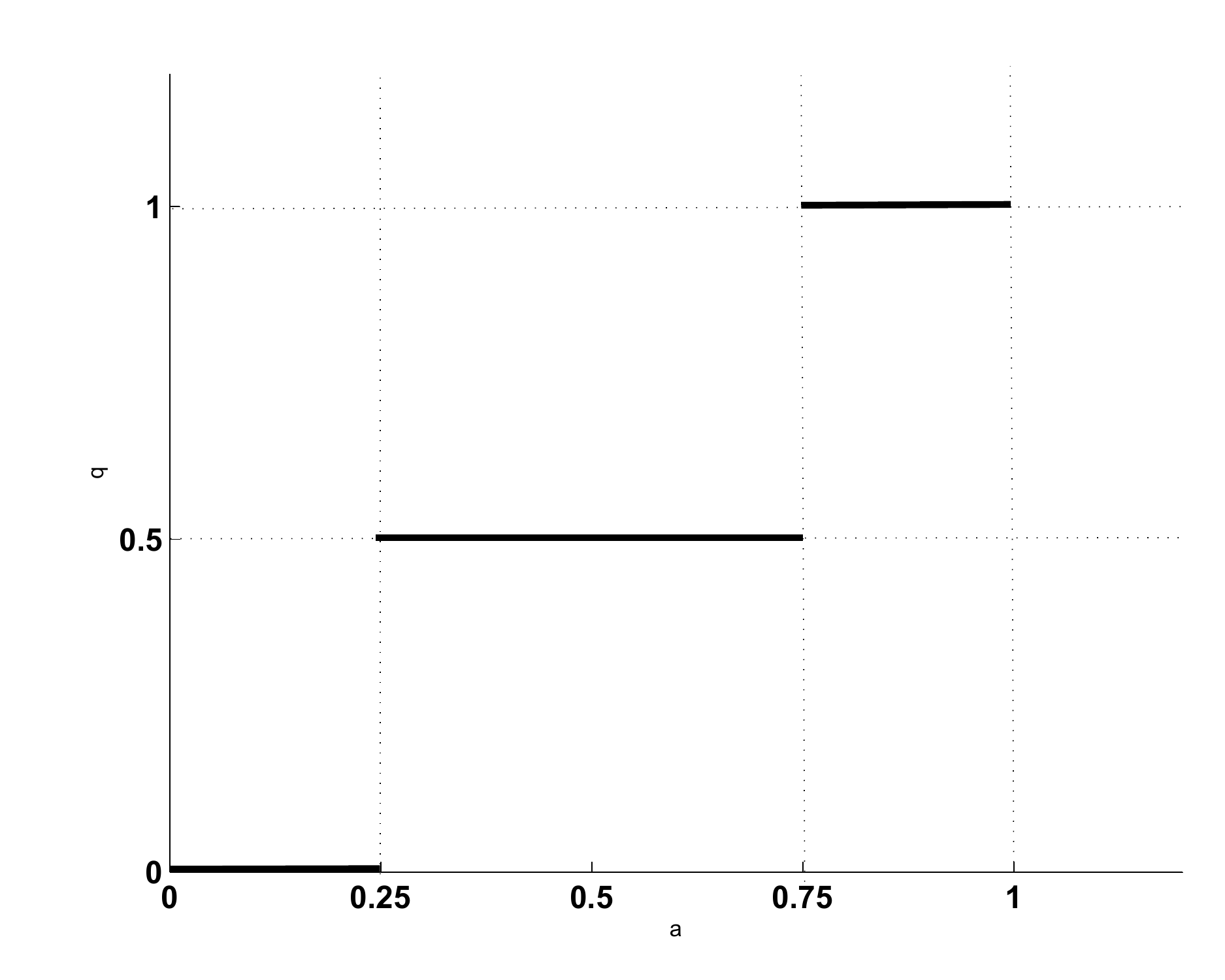}
\caption{This figures shows the truthful mapping in Example \ref{lowerbound}.}
\label{fig:lb}
\end{figure}

Schools of type $A$ have a higher average ability than $B$, so in the first phase employers generally prefer students in $A$ to $B$. In expectation, $A$'s students take lower than average job quality by not participating in early contracting: They take $\frac{1}{2}$ in expectation, where average quality is $\frac{1+\epsilon}{2}$. On the other hand, the jobs that are assigned to students with ability lower than $\frac{1}{2}$, are willing to offer to type $A$ students in the first phase.  Therefore the jobs that were supposed to be mapped to $[\frac{1}{4},\frac{1}{2}]$, will be mapped to school $A$ in the first phase. By removing these jobs and the fraction of students they are mapped with, the matching between the remaining set of jobs changes and there are still jobs with quality $\frac{1}{2}+\epsilon$ will mapped to abilities lower than $\frac{1}{2}$, if don't offer in the first phase. Like before these jobs will offer and to students of type $A$ in the first phase. These mappings continue until there are no quality-$\frac{1}{2}+\epsilon$ jobs mapped to ability lower than $\frac{1}{2}$.

Eventually jobs that are mapped to $[\frac{1}{4},\frac{3}{4}]$ in $Q_T$, are mapped to students in type $A$ in first phase. And other jobs are remained to be mapped to students in school $B$ when the true grades are revealed.

Neglecting $\epsilon$ in the optimal matching, quality-$0$ are mapped to $[0,\frac{1}{4}]$, quality-$\frac{1}{2}$ to $[\frac{1}{4},\frac{3}{4}]$ and quality-$1$ jobs are mapped to $[\frac{3}{4},1]$ abilities. The social welfare in this case is $\frac{11}{32}$.
In the equilibrium, students in school $A$ take job quality $\frac{1}{2}$ in phase 1. In phase 2, quality-$0$ jobs are mapped to $[\frac{1}{4},\frac{1}{2}]$ and quality-$1$ are mapped to $[\frac{1}{2},\frac{3}{4}]$. Therefore the social welfare is $\frac{9}{32}$.


\end{example}

In the following example for early contracting with single type of schools, the ratio between optimal and equilibrium efficiency is $\approx 1.11$. This shows that price of anarchy in this setting is $\geq 0.11$.

\begin{example}\label{ex:1.11}
Suppose the distribution of jobs is as follows: $0.2$ of jobs are quality-$0$, $0.45$ of jobs are quality-$0.64$ and the other $0.35$ fraction are of quality $1$.\footnote{Similar to the previous example, with slight change this distribution fits into distribution properties in section \ref{sec:prelim}. Since it is not going to effect the matching, we continue with the distribution in the example for notational simplicity.} Jobs with $0.64$ are better than average and some of them will get lower than average student if they do not offer in the first phase.

So the early offer starts from jobs mapped to $0.2$ and to form a symmetric interval as discussed in \ref{EqStruct}, this interval continues to ability level $0.8$. After mapping of these jobs that consists of all $0.64$-jobs and $0.15$ quality-$1$ in the first phase, quality-$0$ will be mapped to interval $[0,\frac{1}{2}]$ and other quality-$1$ will be mapped to $[\frac{1}{2},1]$ in the second phase.

The social welfare in the optimum case is $0.45 \times 0.64 \times \frac{0.2+0.65}{2}+ 0.35 \times 1 \times \frac{0.65+1}{2}$ because $0.45$ fraction of students with average $\frac{0.2+0.65}{2}$ take $0.64$ and $0.35$ fraction of students with average $\frac{0.65+1}{2}$ take quality $1$.

However in the equilibrium mapping $0.6$ of the jobs who are a mixture of quality-$0.64$ and quality-$1$ with average $\frac{0.64 \times 0.45 + 1 \times 0.15}{ 0.6}$ are assigned to  students with expected ability $\frac{1}{2}$ in the first phase. In the second phase, $0$-jobs are are assigned to remaining $0.2$ fraction of students with ability less than $\frac{1}{2}$ and remaining $0.2$ fraction of jobs that have quality $1$ are assigned to abilities in $[\frac{1}{2},1]$. So the social welfare in this case is $0.6\times \frac{0.64 \times 0.45 + 1 \times 0.15}{ 0.6} \times \frac{1}{2}+ 0.2 \times 1 \times \frac{0.5+1}{2}$.

The ratio between these two social welfare values is $\approx 1.11$.
\end{example}






\end{document}